\documentclass[12pt,leqno]{article}
\usepackage{amssymb,amsfonts,amsmath,amsthm,amscd}
\usepackage{graphicx}

\newtheorem{thm}{Theorem}[section]
\newtheorem{lem}[thm]{Lemma}
\newtheorem{prop}[thm]{Proposition}
\newtheorem{cor}[thm]{Corollary}
\newtheorem{df}[thm]{Definition}
\newtheorem{rem}[thm]{Remark}

\newtheorem{ass}[thm]{Assumption}

\newcommand{\E}{\mathbb{E}}

\newcommand{\prob}{\mathbb{P}}

\newcommand{\R}{\mathbb{R}}

\newcommand{\N}{\mathbb{N}}

\newcommand{\hilbert}{\mathcal{H}}

\newcommand{\beq}{\begin{equation}}
\newcommand{\eeq}{\end{equation}}
\newcommand{\bea}{\begin{aligned}}
\newcommand{\eea}{\end{aligned}}
\newcommand{\bdm}{\begin{displaymath}}
\newcommand{\edm}{\end{displaymath}}
\newcommand{\barr}{\begin{array}}
\newcommand{\earr}{\end{array}}
\newcommand{\ben}{\begin{enumerate}}
\newcommand{\een}{\end{enumerate}}
\newcommand{\bde}{\begin{description}}
\newcommand{\ede}{\end{description}}

\newcommand{\defi}{\stackrel{\text{def}}{=}}

\newcommand{\be}{\beta}

\newcommand{\de}{\delta}

\newcommand{\s}{\sigma}

\def\be{\begin{equation} }
\def\ee{\end{equation} }

\begin{document}

\noindent

\bigskip
\begin{center}
{\large \bf SMALL PERTURBATIONS \\ OF A SPIN GLASS SYSTEM}
\end{center}
\vspace{0.5cm}

\begin{center}
Louis-Pierre Arguin\footnote[1]{Supported by the NSF grant DMS-0604869.}  \\
\vspace{0.5cm}
Courant Institute of Mathematical Sciences, New York University,\\
251 Mercer St., NY 10012, USA.\\
\vspace{1cm}

Nicola Kistler\footnote[2]{Supported by the Deutsche Forschungsgemeinschaft, no. DFG GZ BO 962/5-3.}\\
\vspace{0.5cm}
Institute of Applied Mathematics, University of Bonn,\\ 
Wegelerstr. 6, DE-53115 Bonn, Germany. 

\end{center}

\bigskip
\begin{abstract}
We show through a simple example that perturbations of the Hamiltonian of a spin glass which cannot be detected at the level of the free energy can completely alter the behavior of the overlap. In particular, perturbations of order $O(\log N)$, with $N\to \infty$ the size of the system, suffice to have ultrametricity emerge in the thermodynamical limit. 
\end{abstract}

\vfill

\newpage 
${}$ \\

\section{Introduction}

By virtue of the seminal works of Guerra \cite{guerra} and Talagrand \cite{talagrand}, the limiting free energy of models of Sherrington-Kirkpatrick(SK)-type is now known to be given by the Parisi Formula. 
However, the purported ultrametric organization of the Gibbs states \cite{parisi}
remains poorly understood.

One piece of evidence for ultrametricity in the SK-type models is obtained through the cavity-dynamics framework of Aizenman, Sims and Starr \cite{aizenman_sims_starr}: one easily checks that the Parisi Formula is obtained when the AS$^2$-functional is evaluated in the Derrida-Ruelle Random Overlap Structures (the {\it ROSt's}), \cite{ruelle}. As these are prototypes of ultrametric structures, the ultrametricity seems very plausible. This however clearly does not imply that the Gibbs measure itself is ultrametric. In \cite{bokis} it is proved that there do exist models whose free energy coincides with that of a hierarchical model but with non-ultrametric Gibbs measure (such models were called {\it non-irreducible}). 

Another piece of evidence in favor of ultrametricity stems from the extended Ghirlanda-Guerra identities (EGGI), especially in view of Panchenko's beautiful result \cite{panchenko}. There are however different problems with the EGGI in relation to ultrametricity. First, it is not known whether the EGGI hold for any temperature, but only ``on average'', cf.\cite{ghirlanda_guerra, talagrand}. Regarding this difficult issue, we have nothing to say. Second, the EGGI are typically obtained by adding small perturbations to the Hamiltonian which leave the free energy of the system unchanged. The fact that this is a somewhat risky endeavour was already clear to Parisi and Talagrand (and presumably to others) who point out in \cite[p. 3]{parisi_talagrand} that \\

{\it ``...to any Hamiltonian one can add a small perturbation term... such that the perturbed Hamiltonian satisfies the EGGI. The perturbation term is small in the sense that it does not change the limiting free energy. (Unfortunately, adding this term might change the structure of the overlap)''.\\}

In this note we address the issue of perturbed Hamiltonians with particular emphasis on ultrametricity. We consider REM-like systems such as those introduced in \cite{bokis} which are {\it not} ultrametric in the thermodynamical limit and show that ``small'' perturbations to the Hamiltonian suffice to have ultrametricity emerge; by this we understand perturbations whose variance is of order $\alpha\log N$, for $\alpha$ large enough and $N$ the size of the system. 

The use of {\it small perturbations} pervades the whole subject of spin glasses, having proved to be crucial in the derivation, e.g., of the Aizenman-Contucci equations \cite{aizenman_contucci}, of the Ghirlanda-Guerra equations \cite{ghirlanda_guerra} and their generalizations EGGI \cite{talagrand}. Usually based on sound stability considerations \cite{contucci_giardina}, small perturbations must however be taken with caution. Indeed, although it is to be expected from general statistical mechanics considerations that the structure of the Gibbs state can be affected by a small perturbation of the Hamiltonian, it is rather surprising that modifications of the order of the logarithm of the size of the system suffice to deeply alter the organization of the states.\footnote{This is to be compared for example to the random field Curie-Weiss model where perturbation of order $N^{1/2}$ are necessary to modify the measure \cite{AZP}.}

\section{General Setting}
Let us start by considering a general Gaussian spin glass system on $N$ spins. Precisely, we take a centered Gaussian process $X=(X_\sigma)_{\sigma\in\Sigma_N}$, $\Sigma_N:=\{-1,1\}^N$,  with covariance or {\it overlap }matrix $Q=N\{q_{\sigma\sigma'}\}$ and $q_{\sigma\sigma}=1$. At this point, we do not specify a form for the overlap matrix $Q$ besides the normalization of the diagonal (and hence no particular geometry of $\Sigma_N$). In our notation, the SK model corresponds to taking $q_{\s\s'}=\left(\frac{1}{N}\sum_{i=1}^N\s_i\s'_i\right)^2$. Throughout the paper we will write $\E$ for the expectation over the process $X$ and $\prob$ for its law. 

The Gibbs measure $\mathcal{G}_{\beta,N}$ on $\Sigma_N$ is defined as usual by
$$
\mathcal{G}_{\beta,N}(\sigma)=\frac{e^{\beta X_\sigma}}{Z_N(\beta)} \ , Z_N(\beta)=\sum_{\sigma\in\Sigma_N}e^{\beta X_\sigma} \ . 
$$
We write $\mathcal{G}_{\beta,N}^{\otimes s}$ for the product measure of $s$ copies of $\mathcal{G}_{\beta,N}$. The free energy is denoted by
$$ 
f_{N}(\beta)\defi \frac{1}{N}\log Z_N(\beta) \ .
$$
We will assume that the limit $N\to\infty$ exists, and that it coincides with the limit of $\E f_{N}(\beta)$ (self-averaging). 

It is useful for our purpose to make sense of the Gibbs measure in the thermodynamic limit $N\to\infty$ (see \cite{arguin} for details). To this aim, one considers the algebra of observables generated by functional of the form
\be
\mathcal{G}_{\beta,N}\mapsto\E\mathcal{G}_{\beta,N}^{\otimes s}\left(\prod_{i<j}^sq^{k_{ij}}_{\sigma_i\sigma_j}\right)
\label{eqn overlap fct}
\ee
for some $k_{ij}\in\N$. Replicas of configurations are denoted by $\sigma_i$. For each $N$, the collection of observables define the law of a weakly exchangeable overlap matrix $Q_{\beta,N}$, i.e., a random matrix whose law is invariant under permutations of rows and columns. 

Weakly exchangeable overlap matrices correspond to Gram matrices constructed by independently sampling vectors from a {\it directing measure} $\mu$ on some canonical Hilbert space $\hilbert$ \cite{ds}. In the above example, the directing measure of $Q_{\beta,N}$ is the Gibbs measure $\mathcal{G}_{\beta,N}$ and the inner product is simply the overlap between configurations. By compactness, one can find a subsequence for which the whole collection of observables converge. Each limiting measure defines a weakly exchangeable covariance matrix, and hence a limiting directing measure, that we refer to as the {\it infinite-volume Gibbs measure} and denote it by $\mathcal{G}_\beta$. We stress that this limit will generally not be unique. The whole set of Gibbs measure of the system is defined to be the closed convex hull of such limit points. By analogy with the finite-volume measure, the pure states at given disorder in this framework correspond to the vectors on which a realization of $\mathcal{G}_\beta$ is supported. We gather these considerations into a proposition.
\begin{prop}
Let $Q_{\beta,N}$ be the overlap matrix constructed by the sampling of the Gibbs measure $\mathcal{G}_{\beta,N}$. Then each limit point of $(Q_{\beta,N})_N$ defines an infinite-volume Gibbs random measure $\mathcal{G}_\beta$ on a canonical Hilbert space $\hilbert$.
\end{prop}
We will denote a generic element of $\hilbert$ by $\sigma$ and the inner product on $\hilbert$ by $q_{\sigma\sigma'}$ to be consistent with the notation of finite systems. We will write $d$ for the distance on $\hilbert$ induced by the inner product
$$ d(\sigma,\sigma')=\sqrt{1-q_{\sigma,\sigma'}}\ .$$
A tantalizing question related to the Gaussian process is to describe the limiting $\mathcal{G}_{\beta}$. For many systems, the organization of the pure states is expected to obey the appealing Parisi Picture.
\begin{df}[Parisi Picture]
A spin glass system is said to satisfy the {\em partial} Parisi Picture if the distance on the support of $\mathcal{G}_{\beta}$ is ultrametric almost surely, i.e.,
$$ \mathcal{G}_{\beta}^{\otimes 3}\left(d(\sigma_1,\sigma_2)\leq \max\{d(\sigma_1,\sigma_3),d(\sigma_2,\sigma_3)\}\right)=1\ .$$
It is said to satisfy the {\em full} Parisi picture if the law of $\mathcal{G}_\beta$ is a Derrida-Ruelle cascade.
\end{df}
The Derrida-Ruelle cascades will be defined below. A first step towards the Parisi picture that can be proven in many examples, and under which the partial picture implies the full one, is the celebrated extended Ghirlanda-Guerra identities
\begin{df}[EGGI]
A Gibbs measure $\mathcal{G}_\beta$ is said to satisfy the extended Ghirlanda-Guerra identities if and only if for all $s\in\N$ and for any bounded measurable function $f:[-1,1]^{s^2}\to\R$ and $g:[-1,1]\to\R$
\begin{equation} \begin{aligned}
& \E\mathcal{G}_\beta^{\otimes s+1}\Big(f(\{q_{\s_i\s_j}\}_{i,j\leq s})g(q_{\s_1\s_{s+1}})\Big)= \\
& \hspace{1cm} = \frac{1}{s}\E\mathcal{G}_\beta^{\otimes s}\Big(f(\{q_{\s_i\s_j}\}_{i,j\leq s})\Big)\E\mathcal{G}_\beta^{\otimes 2}\Big(q_{\s_1\s_2}\Big)+\frac{1}{s}\sum_{l=2}^s\E\mu^{\otimes s}\Big(f(\{q_{\s_i\s_j}\}_{i,j\leq s})g(q_{\s_1\s_l})\Big)\ .
\end{aligned} \end{equation}
\end{df} 
We remark that the identities are non-linear in the law of $\mathcal{G}_\beta$, because of the product appearing on the right-hand side. Therefore the identities cannot hold for convex combination of Gibbs measures, but only for extreme ones.
It is well established that EGGI is a necessary condition for the full Parisi picture to hold \cite{bk2}. 
This leads to the natural question, is EGGI a sufficient condition for ultrametricity ?
A recent and beautiful result of Panchenko shows that it actually is, provided the overlaps can only take a finite number of values.
\begin{thm}[Panchenko] \label{panchy_one}
If a measure $\mathcal{G}$ satisfies EGGI and the number of values taken by the non-diagonal entries is finite, then almost surely
$$\mathcal{G}^{\otimes 3}\Big(d(\sigma_i,\sigma_j)\leq \max\{d(\sigma_i,\sigma_k),d(\s_j,\s_k)\}\Big)=1\ .$$
\end{thm}
Panchenko's theorem thus establishes EGGI as a non-trivial yet simple criteria for a spin glass system to satisfy the Parisi picture. 
A large class of spin glass models, the so-called stochastically stable ones, \cite{aizenman_contucci,contucci_giardina} satisfy the Ghirlanda-Guerra identities when $g(q)=q$ for almost all value of $\beta$. The extended identities are much stronger, since valid for all bounded $g$, and proven in the case of REM-like models \cite{bk2}. For more involved models, like in the SK-type model, and as \cite{ghirlanda_guerra}, one can retrieve EGGI by constructing a perturbed Hamiltonian $X^\delta$ by adding to the original system independent Gaussian fields $(X_\sigma^{p})$ with covariance $N\delta_N\{q_{\s\s'}^p\}$ for integer $p>1$ where $\delta_N\to 0$.
The perturbation is chosen in such a way that: {\it i)} The free energy of the perturbed system is the same as the original one: $f^\delta=f$; {\it ii)} The standard procedure to prove the identities can be applied for each $p$.  The extended identities then hold for all bounded measurable $g(q)$ by approximation. 

The question we address in this paper is motivated by the use of EGGI as a tool to investigate the ultrametricity of the limiting Gibbs measure of the original system:
\begin{quote}
{\it If $X$ and $X^\delta$ are two spin glasses with the same free energy, does $\mathcal{G}^\delta_\beta$ being ultrametric implies so for $\mathcal{G}_\beta$ ? 
}
\end{quote}
In the next section, we provide an example of a simple system for which the answer is no, cf. Theorem \ref{energy_levels} and Corollary \ref{full_parisi_picture}. This in effect also shows that the Gibbs measure is not continuous with respect to the perturbation, cf. Corollary \ref{eggi_perturbed}. 
The procedure we choose will be different from an expansion in $p$-powers of the covariance matrix though equivalent, and {\it ad hoc} to our example. 
This has the advantage of being valid at all temperature as well as providing more insights and better control on the effect of the perturbation. The proofs are postponed to Section \ref{proofs}. For completeness, the method of $p$-power expansion is outlined in an appendix.

\section{Perturbations of Non-Irreducible Spin Glasses}
\subsection{Definition of the Example}
In \cite{bokis} some nonhierarchical versions of Derrida's  GREM were introduced. It was proved that the free energy always coincides in the thermodynamical limit with the free energy of a suitably constructed GREM. On the other hand it was shown in 
 \cite{bokis_two} that not all the systems of the form \cite{bokis} are genuinely ultrametric. Such models were called non-irreducible. We are going to consider here the simplest non-irreducible Hamiltonian.

Let $N\in \N$, and consider $\sigma=(\sigma_1,\sigma_2)\in\Sigma_{N}$ where $\sigma_1,\sigma_2 \in \Sigma_{N/2}$. 
We define the Hamiltonian
\beq \label{Hamiltonian}
X_{\s} \defi X^{(1)}_{\s_1} + X^{(2)}_{\s_2},
\eeq
where $(X^{(1)}_{\s_1})$, $\s_1\in\Sigma_{N/2}$,  are iid centered Gaussians of variance $Na_1$ and so is $(X^{(2)}_{\s_2})$, $\s_2\in\Sigma_{N/2}$, with variance $Na_2$ and  independent of $X^{(1)}$. Here $a_1, a_2$ are positive parameters such that $a_1+a_2 = 1$, and, without loss of generality we assume that $a_1 > a_2$.

By definition, the overlap $q_{\sigma\tau}$ between two distinct configurations $\sigma$ and $\tau$ can only take the values $1$ if $\sigma_1 = \tau_1, \sigma_2 = \tau_2$, $a_1$ if $\sigma_1=\tau_1$, $a_2$ if $\sigma_2=\tau_2$ and $0$ if neither projection of $\sigma$ corresponds. The reader can verify easily that the distance induced by the overlaps is not an ultrametric.

The limiting free energy $f_N(\beta) \defi \lim_{N\to \infty} f_N(\beta)$ of the spin glass \eqref{Hamiltonian} exists, is self-averaging and coincides with that of a two-levels GREM \cite{bokis}. (Our choice $a_1>a_2$ prevents the system from collapsing to a REM.) The Gibbs measure is however clearly a product measure and will remain so in the limit, $\mathcal G_{\beta, N}(\s) = \mathcal G_{\beta, N}^{(1)}\otimes \mathcal G_{\beta, N}^{(2)}$,
with $\mathcal{G}_{\beta, N}^{(1)}$ and $\mathcal G_{\beta, N}^{(2)}$ denoting the first and second marginal respectively. 
Hence, by the structure of the overlaps, such a measure cannot exhibit ultrametricity (unless the trivial one). 
\begin{lem}\label{lem unperturbed}
The support of the limiting Gibbs measure $\mathcal{G}_\beta$ of the system \eqref{Hamiltonian} is not ultrametric. In particular, it does not satisfy EGGI.
\end{lem}
\begin{proof}
The second assertion can be checked directly. It is also a straightforward application of Panchenko's theorem.
\end{proof}
\subsection{The perturbed Hamiltonian}
We now introduce a {\it small perturbation} of \eqref{Hamiltonian}. For a parameter $\de>0$ which will measure the ''strength'' of the perturbation, we consider an additional family of independent centered Gaussians 
$(X^{\de}_{\sigma_1,\sigma_2})$ with variance $N a_2 \delta\ \omega(N)$, where as $N\to\infty$,
\[\bea
\omega(N)\to 0, \ N\omega(N)\to\infty \ .
\eea
\]
\begin{ass} \label{assump_speed}
$N \omega(N)$ tends to $+ \infty$ at least as fast as $\alpha\log N$ for $\alpha>\frac{2}{\log 2}$. 
 \end{ass}
(It will turn out that this speed is, as long as the extremal process is concerned, optimal, in the sense that smaller perturbations leave the 
asymptotical properties of the extremal process unchanged, cfr. Remark \ref{explanations} below.) \\

We set the {\it perturbed Hamiltonian} to be 
\beq 
X_\s^{\de} \defi X^{(1)}_{\s_1} + X^{(2)}_{\s_2} + X_{\s_1, \s_2}^{\de}.
\eeq
We define partition function $Z_{\delta,N}(\beta)$, free energy $f_{\delta,N}(\beta)$ and Gibbs measure $\mathcal G_{\beta, \de, N}$ in the obvious manner. 

The following shows that such a perturbation is indeed {\it small}: 
\begin{lem}
The limit $f_\de(\beta) \defi \lim_{N\to \infty} f_{\delta,N}(\beta)$ exists and is self-averaging. Moreover, for any $\delta>0$, 
\[
f_\delta(\beta) = f(\beta). 
\]
\end{lem}
\begin{proof}
Denoting by $\E_\de$ integration with respect to the $X_{\s_1, \s_2}^\de$-field, it follows by Jensen's inequality that   
\[ \bea
\E f_{\de,N}(\beta) &\leq \E {1\over N} \log  \sum_{\s\in \Sigma_N} \E_\de\Big[ \exp\left[\beta X_\s + \beta X_{\s_1, \s_2}^\de \right] \Big] = \E f_N(\beta) + {\beta^2\over 2} a_2 \de\omega(N). 
\eea \]
Taking the limit $N\to \infty$ gives the upper bound (in expectation). On the other hand, we may rewrite 
\[ \bea
\E f_{\de,N}(\beta) &= \E \log \mathcal{G}_{\beta, N}\Big(\exp \beta X_\s^\de \Big) + \E f_N(\beta) \\
& \geq \E \mathcal{G}_{\beta, N} \Big(\beta X_{\s_1,\s_2}^\de \Big) + \E f_N(\beta) = \E f_N(\beta),
\eea \]
where the inequality follows again by Jensen. This yields the lower bound (in expectation).
The self-averaging follows by concentration of measure, see e.g. Theorem 2.2.4 in \cite{talagrand}, once it is observed that $f_{\delta,N}(\beta)$ has Lipshitz constant smaller than $$\beta N^{-1/2}\sqrt{a_1+a_2(1+o(1))}\ .$$
\end{proof}

\subsection{Gibbs Measure of the Perturbed Hamiltonian}
In order to describe the properties of the Gibbs measure associated to the perturbed Hamiltonian, we need to recall some objects, related to the Derrida-Ruelle cascades. 

Consider the point process  $(\xi_{\boldsymbol i}, \boldsymbol i \in \N^2)$, with $\xi_{\boldsymbol i} \defi \xi_{\boldsymbol i_1}^1 + \xi_{\boldsymbol i_1, i_2}^2$, with the following properties: {$1.$}  $(\xi_{\boldsymbol i_1}^1, i_1 \in \N)$ a Poisson Point Process of density $\beta_1 e^{- \beta_1 t}dt$, with $\beta_1 \defi \sqrt{\log 2 \over a_1}$. {$2.$} For given $i_1$ the Point Process  $(\xi_{\boldsymbol i_1, i_2}^2, i_2 \in \N)$ is Poissonian with density $\beta_2 e^{- \beta_2 t}dt$, with $\beta_2 \defi \sqrt{\log 2 \over a_2}$. {$3.$} For different $i_1, i_1'$, the point processes $(\xi_{\boldsymbol i_1, j}^2, j)$ and $(\xi_{\boldsymbol i'_1, j}^2, j)$ are independent. (Remark that, in virtue of our choice $a_1> a_2$, it holds $\beta_1 < \beta_2$ strictly: this will become important.)

We construct a {\it marked point process} (mPP for short) on $\R^2 \times \{0,a_1\}$ by setting 
\[
\mathcal X_{DR} \defi \sum_{\boldsymbol i \neq \boldsymbol i'} \de_{\xi_{\boldsymbol i}, \xi_{\boldsymbol i'}, q_{\boldsymbol i \boldsymbol i'}},
\]
where the overlap $q_{\boldsymbol i \boldsymbol i'}$ of two multi-indices $\boldsymbol i, \boldsymbol i'$ is defined as $0$ if $i_1\neq i_1'$ and $a_1$ otherwise. Note that by construction the overlaps of $\mathcal X$ define an ultrametric.

In the limit $N\to\infty$, it is convenient to look at the {\it shifted energy levels} $(X_\s^\de - a_N)$ for
\[ \bea
a_N&\defi a_N^{(1)} + a_N^{(2)}(\delta)\\
a_N^{(1)} &\defi N \sqrt{a_1 \log 2} - {a_1\over 2\sqrt{a_1 \log 2}} \log(2\pi a_1 N) , \\
a_N^{(2)} (\delta)&\defi N \sqrt{a_2 (1+\de_N) \log 2} -  {a_2(1+\de_N)\over 2 \sqrt{a_2 (1+\de_N) \log 2}}  \log(2\pi N a_2(1+\de_N)).
\eea\]
where we write $\delta_N=\delta \omega_N$ for short.

The following shows that such a small perturbation can turn a non-ultrametric system such as \eqref{Hamiltonian} into an ultrametric one.
We formulate the result first for the {\it extremal process}.  

\begin{thm}\label{energy_levels}
Under assumption \ref{assump_speed}, and for any $\delta>0$ the mPP of the shifted energy levels $$\mathcal X^\delta_N \defi  \sum \de_{X^\de_\s - a_N, X^\de_\tau - a_N, q(\s, \tau)}$$ converges weakly to $\mathcal X_{DR}$. 
\end{thm}
To relate this result to the Parisi picture, we need to recall the {\it multiplicative} Derrida-Ruelle cascades. By these we understand the image of $\mathcal X$ under the mapping $s\mapsto \exp(\beta s)$, where $\beta > \beta_2$. This is simply the above marked point process with points $\xi_{\boldsymbol i}$ replaced by $\eta_{\boldsymbol i} \defi \exp(\beta \xi_{\boldsymbol i})$, that is 
\[
\mathcal Y \defi \sum_{\boldsymbol i \neq \boldsymbol i'} \de_{\eta_{\boldsymbol i}, \eta_{\boldsymbol i'}, q_{\boldsymbol i \boldsymbol i'}}.
\]
We observe that $\beta> \beta_2$ insures that $\sum \eta_i < \infty$ almost surely.  By $\mathcal Z$ we understand the {\it normalized} Derrida-Ruelle multiplicative cascades, namely the above Point Process where the points $\eta_{\boldsymbol i}$ are replaced by their normalized counterparts $\overline \eta_{\boldsymbol i} \defi \eta_{\boldsymbol i}/\sum_{\boldsymbol j} \eta_{\boldsymbol j}$.

The normalized cascade is nicely expressed in terms the Bolthausen-Sznitman coalescent, introduced in \cite{boszni_ruelle}. This is a continuous time Markov
process $(\psi_t, t\geq 0)$ taking values in the compact set of partitions on $\N$. We call a partition $\mathcal C$ finer than $\mathcal D$, in notation ${\mathcal C} \succ \mathcal
D$, provided that the sets of $\mathcal D$ are unions of the sets of $\mathcal C$. The process $(\psi_t, t\geq 0)$ has the following properties: {\it i.} If $t\geq s$  then $\psi_s \succ \psi_t$. 
{\it ii.} The law of $(\psi_t, t\geq 0)$ is invariant under permutations involving only a finite number of elements. {\it iii.} $\psi_0 = \{\{1\},\{2\},...\}$. We denote the equivalence relation associated with $\psi_t$ by $\sim_t$. To every pair of point corresponds a stopping time $t(i,j):=\min\{l: i \sim_{t} j\}$. 

For $x_l := \beta_l/\beta$, $l=1,2$, we pick $t_0 = 0 < t_1 < t_2 < \infty$ with $t_l = \log(x_2/x_{2-l})$. The overlap $q_{ij}$ is defined to be $0$ if $t(i,j)>t_2$, $a_1$ if $t_2>t(i,j)>t_1$ and $1$ otherwise.
Given a Poisson Point Process $(z_i, i \in \N)$, one can construct a marked point process \cite{bokis_two}
$$\sum_{i \neq i'} \de_{z_i, z_{i'}, q_{ii'}}$$
where the marks $q_{ii'}$ are chosen randomly as above, independently of the point process $(z_i, i \in \N)$. The law of such an object is denoted by $P \sqcap \mathcal C$, where $P$ is the law of the underlying point process, and $\mathcal C$ that of the coalescent. A normalized cascade can be shown to have the law $P_x \sqcap \mathcal C$ where $P_x$ is the law of the normalization of the Poisson point process with 
density $x t^{-x -1} dt$ on $\R_+$ \cite{boszni_ruelle}.

\begin{cor}[Full Parisi Picture] \label{full_parisi_picture}
Let $\beta > \beta_2$. Then the marked point process of the Gibbs measure associated to the perturbed Hamiltonian
\[
\sum_{\s \neq \tau} \de_{\mathcal G_{\beta, \de, N}(\s), \mathcal G_{\beta, \de, N}(\tau), q(\s, \tau)}
\]
converges weakly towards $P_{x_2} \sqcap \mathcal C$, where $P_{x_2}$ is the law of the normalization of the poisson point process with 
density $x_2 t^{-x_2 -1} dt$ on $\R_+$. 
\end{cor}

A direct consequence of the above is that the EGGI hold for the perturbed Hamiltonian, thereby proving the discontinuity of the Gibbs state under the perturbation.
\begin{cor}[Perturbed Hamiltonian and EGGI] \label{eggi_perturbed} 
The limiting Gibbs measure of the perturbed system $\mathcal{G}_{\beta,\delta}$ satisfies EGGI. In particular, in the sense of the topology induced by the functions \eqref{eqn overlap fct},
$$ \lim_{\delta\to 0} \mathcal{G}_{\beta,\delta}\neq \mathcal{G}_\beta$$
where $\mathcal{G}_\beta$ is the limiting Gibbs measure of the original system \eqref{Hamiltonian}.
\end{cor}
\begin{proof}
By Corollary \ref{full_parisi_picture}, $\lim_{\delta \to 0} \mathcal{G}_{\beta, \delta}$ is a Derrida-Ruelle cascade. They are well-known to satisfy EGGI, cfr. Bovier and Kurkova's work \cite{bk3}. By Lemma \ref{lem unperturbed}, $\mathcal{G}_\beta$ does not. The conclusion follows from the fact that the identities are continuous in the topology determined by the observables \eqref{eqn overlap fct}. 
\end{proof}

\section{Proofs}\label{proofs}
The proofs of Theorem \ref{energy_levels} and of Corollary \ref{full_parisi_picture} very closely follow the line of proof of the Main Theorem in \cite{bokis_two}. To keep this work reasonably self-contained we shall however outline the crucial steps, especially those steps which differ from the analysis in \cite{bokis_two}. (It turns out that these differences are only very small.) 

Throughout, $K$ will denote a constant, not necessarily the same at different occurrences.  
We shorten the notation for the shifted processes $\hat{X}^1_{\sigma_1} \defi X_{\sigma_1}^1 - a_N^{(1)}$, 
$\hat{X}^2_{\sigma_2} \defi X_{\sigma_2}^2 - a_N^{(2)}(\delta)$ and $\hat{X}^\delta_\s\defi \hat{X}_{\sigma_1}^1+\hat{X}^2_{\sigma_2}+X_{\s_1,\s_2}^\delta$. 
We will need the following straightforward asymptotics
\begin{equation} \bea \label{asymp}
& {a_N^{(1)}\over a_1 N} =  \beta_1 + O\left({ \log N\over N}\right), \quad \exp\left[- {{a_N^{(1)}}^2 \over 2 a_1 N }\right] = 2^{-N/2} \beta_1 \sqrt{2\pi a_1 N}(1+o(1)). \\
& {a_N^{(2)}(\delta)\over a_2 N(1+\de_N)} = \beta_2 (1+O(\de_N)), \\
& \hspace{2cm} \exp\left[- {{a_N^{(2)}}^2 (\delta)\over 2 a_2(1+\de_N) N }\right] = 2^{-N/2} \beta_2 \sqrt{2\pi a_2(1+\de_N) N}(1+o(1)).
\eea \end{equation}

\begin{lem} \label{localization}
Let $M$ be a compact set. For given $\epsilon>0$ there exists large enough compact $\tilde M$ such that
\[ \bea 
\prob\Big[\exists \s\in \Sigma_N, \text{such that}\,  \hat{X}^\delta_{\s} \in M,\,\text{but}\,
\hat{X}^{1}_{\s_1}\notin \tilde M\,\text{or}\,\hat{X}_{\s_2}^{(2)} +X_{\s_1, \s_2}^\de  \notin \tilde M  \Big] \leq \epsilon
\eea \] 
for large enough $N$.
\end{lem}
\begin{proof}
We first claim that to $\varepsilon >0$ there exists $C>0$ such that 
\beq \label{bounded_one}
\prob\left[\exists \s_1\in\sigma_{N/2}:\; \hat{X}^{1}_{\s_1}\geq C \right] \leq \varepsilon.
\eeq
This is straightforward: the left side of the above expression is bounded by 
\[
2^{N/2}\prob\left[ \hat{X}^{1}_{\sigma_1} \geq C\right] \leq K e^{-\beta_1 C},
\]
where the second inequality follows from the asymptotics \eqref{asymp}. It thus suffices to choose $C$ large enough in the positive.\\

We now claim that to $ \epsilon>0$ there exists $R>0$ such that
\beq \bea \label{bounded_three} 
& \prob\Big[\exists \s\in \Sigma_N\, \text{such that}\, \hat{X}_\s^\de \in M,\, \text{but}\, \hat{X}_{\s_1}^{(1)} \notin [-R, R]\Big] \leq  \epsilon.
\eea \eeq
By \eqref{bounded_one}  we can find $\hat R$ large enough in the positive such that 
\beq \label{intermediate}
\prob\left[\exists\; \s_1: \, \hat{X}^{(1)}_{\s_1} \geq  \hat R \right] \leq \epsilon/2. 
\eeq
On the other hand, 
\beq \bea  \label{bounded_four}
& \prob\Big[\exists \s\in \Sigma_N:\, \hat{X}_\s^\de \in M, \hat{X}_{\s_1}^{(1)}\leq -\tilde R\Big]\leq 2^N \prob\Big[\hat{X}^\delta_\sigma \in M, \, \hat{X}_{\s_1}\leq - \tilde R  \Big] \\
& \leq   
2^N\ \E\left[ \int_{M- \hat{X}^1_{\s_1}} \exp\left[- \frac{(y+a_N^{(2)}(\delta))^2}{ 2 a_2 N(1+\de_N)} \right]\frac{dy}{\sqrt{2\pi a_2 N(1+\de_N)}}  ; \hat{X}^1_{\s_1}\leq -\tilde R\right].
\eea \eeq
Omitting the positive terms in the expansion of the quadratic polynomial we have 
\beq
\exp\left[ - {\left(y+a_N^{(2)}(\delta) \right)^2\over 2 a_2 N(1+\de_N)} \right]\leq  \exp\left[- {{a_N^{(2)}(\delta)}^2 \over 2 a_2(1+\de_N) N } - {a_N^{(2)}(\delta) \over 2a_2N(1+\de_N)}y \right]
\eeq
which, by the asymptotics \eqref{asymp}, is
\beq \bea 
\leq K 2^{-N/2} \sqrt{2\pi a_2 N(1+\de_N)} \exp[-\beta_2 y]\ .
\eea \eeq 
We thus obtain 
\beq 
\eqref{bounded_four} \leq K 2^{N/2} \ \E\left[ \exp\left(\beta_2 \hat{X}^1_{\s_1} \right);    \hat{X}^1_{\s_1} \leq - \tilde R \right]\int_M e^{-\beta_2 y} dy .
\eeq
It is straightforward to see that $\E\left[ \exp\left(\beta_2 \hat{X}^1_{\s_1} \right); \hat{X}^1_{\s_1} \leq - \tilde R \right] \leq K 2^{-N/2} e^{-(\beta_2 - \beta_1 ) \tilde R}$. Combining, we have 
\[
\prob\Big[\exists \s\in \Sigma_N:\, \hat{X}_\s^\de\in M, \hat{X}_{\s_1}^{(1)}\leq -\tilde R\Big] \leq K \exp\left(- (\beta_2 - \beta_1 ) \tilde R \right),
\]
and since $\beta_2 - \beta_1 > 0$, it suffices to choose $\tilde R$ large enough in the positive to make the above smaller than $\epsilon/2$: this then yields \eqref{bounded_three} with $R := \max(\hat R, \tilde R)$. \\

Now, $\hat{X}_\s \in M$ and $\hat{X}_{\s_1}^{(1)} \in [-R, R]$ implies that $\hat{X}_{\s_2}^{(2)} +X_{\s_1, \s_2}^\de \in M-[-R, R]$. The claim of the Lemma thus follows with $\tilde M$ chosen large enough to contain both $[-R, R]$ and $M-[-R, R]$. 
\end{proof}

The following Lemma provides the crucial piece of information pertaining the ultrametricity of the perturbed system. We emphasize that the statement is wrong if $\de= 0$, that is when the Hamiltonian is simply 
$X_\s = X_{\s_1}^1+X_{\s_2}^{2}$: in that case, coincidence of two configurations on the second spin does {\it not} imply also equality on the first.  
\begin{lem} \label{propa}
Let $M$ be a compact set and $\varepsilon>0$. Then
\[
\prob\left[\exists\, \s, \tau\in \Sigma_N: \s_1 \neq \tau_1,\, \s_2 = \tau_2\,\text{such that}\,  \hat{X}_\s^\de,\  \hat{X}_\tau^\de \in M \right] \leq \varepsilon,
\] 
for large enough $N$.
\end{lem}

\begin{proof}
 By Lemma \ref{localization} we can find compact $\tilde M$ such that 
\[ \bea 
\prob\Big[\exists \s\in \Sigma_N, \text{such that}\,  \hat{X}_{\s}  \in M,\,\text{but}\,\hat{X}^{1}_{\s_1} \notin \tilde M\,\text{or}\,\hat{X}_{\s_2}^{(2)} +X_{\s_1, \s_2}^\de \notin \tilde M  \Big] \leq {\varepsilon\over 2}.
\eea \] 
Thus, 
\beq \bea \label{propagation}
& \prob\left[\exists\, \s, \tau\in \Sigma_N, \s_1 \neq \tau_1,\, \s_2 = \tau_2:  \hat{X}_\s^\de, {\hat X}_\tau^\de\in M \right]\leq {\varepsilon\over 2}+\\
&  \prob\Big[\mathop{\bigcup_{\s, \tau\in \Sigma_N}}_{\s_1 \neq \tau_1,\, \s_2 = \tau_2}\left\{
\hat{X}_{\s_1}^{1},\hat{X}_{\tau_1}^{1}\in \tilde M \text{ and } \hat{X}^2_{\s_2}+X_{\s_1, \s_2}^\de, \hat{X}^2_{\sigma_2}+X_{\tau_1, \s_2}^\de\in \tilde M\right\}\Bigg] \\
& \leq {\varepsilon\over 2} + 2^{3N/2}
\prob\left[\hat{X}^1_{\s_1}\in \tilde M\right]^2 \prob\left[\hat{X}^2_{\s_2}+X_{\s_1, \s_2}^\de \in\tilde M ,\hat{X}^2_{\s_2}+X_{\tau_1, \s_2}^\de \in\tilde M \right]
\eea \eeq 
since for $\s_1\neq \tau_1$ the random variables $\hat{X}_{\s_1}^{1}, \hat{X}_{\s_1, \s_2}^\de$ and 
$\hat{X}_{\tau_1}^{1}, \hat{X}_{\tau_1, \tau_2}^\de $ are independent. 

Now it is easily checked that
$$\prob\left[\hat{X}^1_{\s_1}\in \tilde M\right]^2\leq K 2^{-N}\ .$$
To prove the assertion we thus need to check that the second probability is of order $2^{-N/2}o(1)$.
We set 
$$ \Delta a^{(2)}_N:=a_{N}^{(2)}(\delta)-a_{N}^{(2)}(0)$$
where we omit the dependence in $\delta $ for simplicity.
And by expanding,
\[\bea 
\Delta a^{(2)}_N=-N\delta_N\left(\frac{\sqrt{a_2\log 2}}{2}+o(1)\right)\ .
\eea \]
We will need the asymptotics
\[\bea 
&\frac{\Delta a^{(2)}_N}{a_2 N\delta_N}=-\frac{\beta^2}{2}+o(1),\ \frac{{\Delta a^{(2)}_N}^2}{2a_2 N\delta_N}=N\delta_N\left(\frac{\log 2}{8} +o(1)\right)\ .
\eea \]
Therefore, for any $x\in\R$
\[\bea \label{eqn estimate}
 \prob\left[\hat{X}^2_{\s_2}+X_{\s_1, \s_2}^\de \in\tilde M - x\right]&=\int_{\tilde{M}-x}\exp\left[-\frac{(y-\Delta a^{(2)}_N)^2}{2\pi a_2N\delta_N}\right]\frac{dy}{\sqrt{2a_2N\delta_N}}\\
& \leq K\frac{e^{-N\delta_N\frac{\log 2}{8} }}{\sqrt{N\delta_N}}
\int_{\tilde{M}-x} \exp\left[-\left(\beta_2/2+o(1)\right)y \right] dy
\\
&= K\frac{e^{-N\delta_N\frac{\log 2}{8} }}{\sqrt{N\delta_N}} \exp\left[\left(\beta_2/2+o(1)\right)x\right]
\eea \]
where the last equality comes from a change of variable.

The second probability in \eqref{propagation} is for $\sigma_1\neq \tau_1$
\beq \bea \label{eqn prob square}
&\prob\left[\hat{X}^2_{\s_2}+X_{\s_1, \s_2}^\de \in\tilde M ,\hat{X}^2_{\s_2}+X_{\tau_1, \s_2}^\de \in\tilde M \right]=\\
& \int_\R \prob\left[X_{\s_1, \s_2}^\de + \Delta a^{(2)}_N \in\tilde M-x \right]^2
\exp\left[-\frac{(x+a_N^{(2)}(0))^2}{2a_2N}\right]\frac{dx}{\sqrt{2\pi a_2N}}\ .
\eea \eeq
The estimate \eqref{eqn estimate} together with the asymptotics for $a_N^{(2)}(0)$ yields the upper bound
\[\bea
& K\frac{e^{-N\delta_N\frac{\log 2}{4}}}{N\delta_N}\int_\R \exp\left[(\beta_2+o(1))x-\frac{(x+a_N^{(2)}(0))^2}{2a_2N}\right]\frac{dx}{\sqrt{2\pi a_2N}}\\
& \leq 2^{-N/2}K\frac{e^{-N\delta_N\frac{\log 2}{4}}}{\delta_N\sqrt{N}} \int_\R \exp\left[o(1)x\right]\frac{e^{-x^2/2a_2N}dx}{\sqrt{2\pi a_2N}}\\
& = 2^{-N/2} K\exp\left[-\frac{N\delta_N}{2}\left(\frac{\log 2}{2}+\frac{\log (N\delta^2_N)}{N\delta_N}+o(1)\right)\right] .
\eea \]
where the last equality follows by integration. It remains to prove that the exponential term tends to $0$. But this is so if $N\delta_N$ is at least of order $\alpha\log N$ for $\alpha>\frac{2}{\log 2}$  since
$$ \frac{\log (N\delta^2_N)}{N\delta_N} =-\frac{\log N}{N\delta_N}+\frac{2\log (N\delta_N)}{N\delta_N}= -\frac{\log N}{N\delta_N}+o(1)\ .$$
\end{proof}
\begin{rem} \label{explanations}
We stress that the above result essentially stands due to the fact that the perturbation introduces a square in the probability \eqref{eqn prob square}. It is quite remarkable that, even for such small perturbations, this alone is enough to make the probability of the event negligible. On the other hand, as long as the {\it extremal process} is concerned, we believe that our Assumption \ref{assump_speed} on the size of the perturbation is fairly optimal, in the sense that smaller perturbations ($o(log N)$, for $N\to \infty$) will presumably force the probability of such an event to stay macroscopic. 
\end{rem}

{\bf Sketch of the Proof of Theorem \ref{energy_levels}}. The content of Lemma \ref{propa} is that one cannot find two configurations $\s, \tau$ with shifted energy levels falling into a prescribed subset for which the overlap $q(\s, \tau) = a_2$: if coincidence on the second spin, then also automatically on the first, whence the two configurations must coincide. But this entails that the configurations falling into prescribed subsets have the same kind of dependencies (= hierarchical) as if they were coming from a two-levels GREM, and it is therefore not surprising that the mPP of the $\de$-perturbed Hamiltonian converges weakly to the one constructed outgoing from a GREM: it is explained in \cite{bokis_two} how this simple observation, together with the asymptotics used above, e.g.
\[\bea 
& \prob\left[\hat{X}_{1} \in M_1, \hat{X}_{1, 1}^\de \in M_2 \right] = 2^{-N}(1+o(1)) \int_{M_1} e^{-\beta_1 y} dy \int_{M_2} e^{-\beta_2 y} dy,
\eea \]
which holds for any compacts $M_1, M_2 \subset \R$, allows to prove that the mPP $\mathcal X^\delta_N$ converges weakly to $\mathcal X$. We will not reproduce the proof here: it is a natural modification of what is known as the Chen-Stein method \cite{barbour} to prove Poisson Approximation.  \\

{\bf Sketch of the proof of Corollary \ref{full_parisi_picture}.} This is rather straightforward. One first proves convergence of 
of the ''image'' of the process $\mathcal X_N$ under the mapping $s\mapsto \exp(\beta s)$. This is very standard: the upshot is that 
\[
\mathcal Y_N \defi \sum \de_{\exp\beta\hat{X}^\de_\s , \ \exp\beta\hat{X}^\de_\tau , \ q(\s, \tau)}
\]
converges weakly towards $\mathcal Y$. \\

Having proved this, it suffices to prove that the normalization  
\[ 
\exp\beta\hat{X}_\s^\de  \mapsto { \exp\hat{X}_\s^\de \over \sum_\tau \exp\beta\hat{X}_\tau}
\]
commutes with the limit $N\to \infty$ to obtain that 
\[
\mathcal Z_N \defi  \sum \de_{ \overline{\exp\beta\hat{X}^\de_\s}, \overline{\exp\beta\hat{X}^\de_\tau },\; q(\s, \tau)},  \qquad \overline{\exp\beta\hat{X}^\de_\s} \defi {\exp\beta\hat{X}^\de_\s \over \sum_\tau \exp\beta\hat{X}^\de_\tau } 
\]
(this is nothing but $\mathcal G_{\beta, \de, N}(\s)$) converges to $\mathcal Z$. Again, this is very standard, and we refer the 
reader to \cite[pp. 34-35]{bokis_two} for the proof that the two operations commute. \\

The Corollary \ref{full_parisi_picture} then follows by the remarkable properties of the Derrida-Ruelle cascades and the coalescent. We refer the reader to \cite[Section 9.2, pp. 93-94]{boszni} for the heuristics, which in particular clarifies how the special form of the intensity of the Point Processes $t\mapsto \exp(- \beta_i t), i=1, 2$, plays a crucial r\^ole. 

$\hfill \square$

\appendix 

 \section{EGGI for General Perturbed Systems}
In this section, we outline the method of perturbation by expansion in $p$-powers. This is done in such a way to leave the free energy unchanged and retrieve the extended Ghirlanda-Guerra identities for almost all values of the parameters. We follow closely the treatment of the $p$-spin model in \cite{talagrand}. The interest of such a method is that, coupled to Panchenko's theorem, it provides a way to prove the Parisi picture for the perturbed Gibbs measure of a fairly wide class of Hamiltonians (see Proposition \ref{prop full general} below).
 
Let $X=(X_\sigma)_{\sigma\in\Sigma_N}$ be a spin glass Hamiltonian with covariance $N\{q_{\sigma\sigma'}\}$ as in the general setting of Section 2. Consider $(\beta_p)_{p\geq1}$ with $\beta_p>0$ and $\sum_{p\geq1}\beta^2_p<\infty$. We write $\vec{\beta}$ for a vector $(\beta_1,\beta_2,\beta_3,...)$. It is convenient to assume that for all $p$: $\beta_p\leq \beta_1$.  We say that a property holds for {\it almost all in $\vec{\beta}$} for the measure given by the product of the Lebesgue measures on $[0,\beta_1]$. The perturbed Hamiltonian is 
\be \label{eqn p expansion}
  \beta_1X_\sigma +\sqrt{\delta_N}\sum_{p>1}\beta_p X_\sigma^{p}
\ee
where $(X^{p}_\s)$ are centered Gaussians with covariance $N\{q_{\s\s'}^p\}$ independent for distinct $p$ and $X$. We shall need that $\delta_N\to 0$ and $N\delta_N^{1/8}\to\infty$. Therefore $N\delta_N$ must grow faster than $N^{7/8}$, a condition much stronger than $\log N$. The application of Panchenko's theorem proven here is:
\begin{prop}[Full Parisi Picture]
\label{prop full general}
Suppose that the number of values taken by the overlaps $\{q_{\sigma\sigma'}\}$ is uniformly bounded in $N$. Then for almost all $\vec{\beta}$, the limit points of $(\mathcal{G}_{\vec{\beta},N})_N$ are Derrida-Ruelle cascades. 
\end{prop}
By Panchenko's theorem, the proof reduces to show EGGI. 
  \begin{lem}
 For almost all $\vec{\beta}$, the limit points of $(\mathcal{G}_{\vec{\beta},N})_N$ satisfy EGGI.
  \end{lem}
The first ingredient is the self-averaging of the internal energy in $\beta$-average coming from convexity and concentration of measure.
 \begin{thm}[Theorem 2.12.1 in \cite{talagrand}]
 In the setting of \eqref{eqn p expansion}, one has for every $p>1$
 $$ \int_{[0,\beta_1]^\N} \E \mathcal{G}_{\vec{\beta},N}\Big(\Big|{X^p_\sigma}/N-\E \mathcal{G}_{\vec{\beta},N}\big({X^p_\sigma}/N\big) \Big| \Big)d\vec{\beta}\leq \frac{K}{N^{1/4}\delta_N^2}$$
 for some constant $K$ independent of $N$ and $p$. For $p=1$, the above holds without $\delta_N$.
 \end{thm}

We denote by $\mathcal{G}_{\vec\beta}$ a generic limit point of $(\mathcal{G}_{\vec{\beta},N})$. We write $f_s(q)$ for any bounded measurable function of the overlaps of $s$ copies. The above theorem is applied directly to prove the factorization essential to the proof of the EGGI. Namely, 
if $N^2\delta^{1/4}\to 0$ as $N\to\infty$, then for every $p$ and almost all $\vec{\beta}$
 $$  \lim_{N\to\infty}\E\mathcal{G}_{\vec{\beta},N}^{\otimes s}\Big(\frac{X^p_{\s_1}\ f_s(q)}{\delta_NN}\Big)=
 \beta_p^2  \ \E\mathcal{G}_{\vec{\beta}}^{\otimes 2}\Big(1-q^p_{\sigma_1\sigma_2}\Big)
  \E\mathcal{G}_{\vec{\beta}}^{\otimes s}\Big(f_s(q)\Big)\ .
 $$
 On the other hand, standard Gaussian integration by parts yields for every $N$
 $$
  \E\mathcal{G}^{\otimes s}_{\vec{\beta},N}\Big(\frac{X^p_{\sigma_1}\ f_s(q)}{\delta_NN}\Big)=
 \beta_p^2\left(\sum_{l=1}^s\E\mathcal{G}^{\otimes s}_{\vec{\beta},N}\Big(q^p_{\sigma_1\sigma_l}\ f_s(q)\Big)-s \ \E\mathcal{G}^{\otimes s+1}_{\vec{\beta},N}\Big(q^p_{\sigma_1\sigma_{s+1}}\ f_s(q)\Big)\right) \ .
 $$

 By combining the two last equations, one gets an approximation of  any bounded measurable function $g$ by approximating with polynomials, thereby retrieving EGGI and proving the proposition. \\
${}$ \hfill $\square$ \\
 
{\large \bf Acknowledgments.} The idea of small perturbations of non-irreducible models was mentioned by Erwin Bolthausen (apparently after  a discussion with Silvio Franz) to the second named author long ago. We thank Bolthausen and Franz for sharing their insights.

\end{document}